\documentclass[10pt]{article}
\usepackage{amsmath}
\usepackage{amsfonts}
\usepackage{amssymb}
\usepackage{mathrsfs}
\usepackage{amsthm}
\usepackage{setspace}
\usepackage{fancyhdr}

\usepackage{longtable}

\usepackage{float}
\usepackage{subfig}
\usepackage{fancybox,graphicx}
\usepackage{subfig}
\usepackage{caption}
\usepackage{color}
\usepackage{authblk}

\usepackage[colorlinks, linkcolor=red, citecolor=blue]{hyperref}

\usepackage{accents}
\usepackage[titletoc,title]{appendix}
\usepackage{cite}

\usepackage{geometry}
\title{Several new  classes of optimal ternary cyclic codes
	with \\two or three zeros}

\newtheorem{theorem}{Theorem}
\newtheorem{lemma}{Lemma}
\newtheorem{corollary}{Corollary}

\newtheorem{remark}{Remark}

\newtheorem{example}{Example}

\newcommand{\RNum}[1]{\uppercase\expandafter{\romannumeral #1\relax}}

\newcommand{\F}{\ensuremath{\mathbb F}}

\newcommand{\done}{\hfill $\Box$ }

\newcommand{\notequiv}{{\,\not\equiv\, }}

\newcommand{\fp}{{\mathbb F}_{3}}

\newcommand{\fthree}{{\mathbb F}_{3}}

\author{
Gaofei Wu
\thanks{G. Wu and Z. You are with the State Key Laboratory of Integrated Service Networks, School of Cyber Engineering, Xidian University, Xi'an, 710071, China.
G. Wu is also with the Advanced Cryptography and System Security Key Laboratory of Sichuan Province. 
Email: gfwu@xidian.edu.cn, martingale2002@163.com. }
, Zhuohui You, 
Zhengbang Zha
\thanks{Z. Zha is with the College of Mathematics and System Sciences, Xinjiang University, Urumqi 830046, China.
Email: zhazhengbang@163.com.}
, 
and Yuqing Zhang
\thanks{Y. Zhang is with  the   National Computer Network Intrusion Protection Center, University of Chinese Academy of Sciences, Beijing 101408, China. Email: zhangyq@ucas.ac.cn.}
}
\date{}

\begin{document}
	
	\maketitle
	
	\begin{abstract}

	Cyclic codes are a subclass of linear codes and
	have wide applications in data storage systems, communication systems and consumer electronics
	due to their efficient encoding and decoding algorithms.
	Let  $\alpha $ be  a generator of $\F_{3^m}^*$, where 
	$m$ is a positive integer. Denote by 
	$\mathcal{C}_{(i_1,i_2,\cdots, i_t)}$
	the  cyclic code with generator polynomial $m_{\alpha^{i_1}}(x)m_{\alpha^{i_2}}(x)\cdots m_{\alpha^{i_t}}(x)$,
	where  ${{m}_{\alpha^{i}}}(x)$ is the minimal polynomial of ${{\alpha }^{i}}$ over  ${{\mathbb{F}}_{3}}$.
	In this paper, by analyzing
	the solutions of certain equations over finite fields,  we present four classes of optimal ternary  cyclic codes $\mathcal{C}_{(0,1,e)}$ and $\mathcal{C}_{(1,e,s)}$
	with parameters $[3^m-1,3^m-\frac{3m}{2}-2,4]$, where $s=\frac{3^m-1}{2}$.
	In addition,
	by     determining the solutions of certain equations  and    analyzing the irreducible factors of certain polynomials over $\F_{3^m}$,
	we  present four classes  of
	optimal ternary cyclic codes  $\mathcal{C}_{(2,e)}$ and  $\mathcal{C}_{(1,e)}$ with parameters $[3^m-1,3^m-2m-1,4]$.
	We show that our new optimal cyclic codes are  inequivalent to the known ones.\\
	\\
	{\bf Index Terms } finite fields, linear codes, minimum distance, cyclic codes.
	\end{abstract}

	\section{Introduction}

	Cyclic codes are a very important subclass of linear codes
	and have been extensively studied.  Throughout this paper,
	let ${{\mathbb{F}}_{{{3}^{m}}}}$ denote the  finite field with ${{3}^{m}}$ elements.
	An $[n,k,d]$ linear code over $\F_3$ is a $k$-dimensional subspace of $\F_3^n$ with minimum Hamming distance $d$.
	An $[n,k]$ \emph{cyclic} code $ \mathcal{C}$ is an $[n,k]$ linear code with the property that any cyclic shift of a codeword is another codeword of $\mathcal{C}$.
	Let $\gcd(n,3)=1$. By identifying any codeword $(c_0,c_1,\cdots, c_{n-1})\in \mathcal{C}$ with
	$$
	c_0+c_1x+c_2x^2+\cdots+c_{n-1}x^{n-1}\in \fp[x]/(x^n-1),
	$$
	any cyclic code of length $n$ over $\fthree$ corresponds to an ideal of the polynomial ring $\fthree[x]/(x^n-1)$.
	Notice  that every ideal of $\fthree[x]/(x^n-1)$ is principal. Thus, any cyclic code can be expressed as
	$\langle g(x)\rangle$, where $g(x)$ is monic and has the least degree. The polynomial $g(x)$ is called the \emph{generator polynomial}
	and $h(x)=(x^n-1)/g(x)$ is called the \emph{parity-check polynomial} of $\mathcal{C}$.
	Let $\alpha$ be a generator   of  ${{\mathbb{F}}_{{{3}^{m}}}^*}$ and
	let ${{m}_{\alpha^{i}}}(x)$ denote the minimal polynomial of ${{\alpha }^{i}}$ over  ${{\mathbb{F}}_{3}}$.
	We denote by $\mathcal{C}_{(i_1,i_2,\cdots, i_t)}$
	the  cyclic code with generator polynomial ${{m}_{{{\alpha}^{{{i}_{1}}}}}}(x){{m}_{{{\alpha}^{{{i}_{2}}}}}}(x)\cdots{{m}_{{{\alpha}^{{{i}_{t}}}}}}(x)$.
	In 2005, Carlet, Ding, and Yuan  \cite{carlet2005linear}
	constructed some
	optimal ternary cyclic codes $ \mathcal{C}_{(1,e)}$ with parameters $[3^m-1,3^m-2m-1,4]$ by using perfect
	nonlinear monomials $x^e$.
	In 2013,
	Ding and Helleseth
	\cite{ding2013optimal} constructed  several
	classes of optimal ternary cyclic codes $\mathcal{C}_{(1,e)}$ by  using monomials including  almost
	perfect   nonlinear  (APN)  monomials.
	Moreover, they  presented  nine open problems on optimal  ternary cyclic codes $\mathcal{C}_{(1,e)}$.
	Open problems 7.5 and 7.8 were solved in
	\cite{li2015conjecture} and \cite{li2014optimal}, respectively.  
	In \cite{li2014optimal}, Li et al.  also presented several classes
	of optimal ternary cyclic codes with parameters $[3^m-1,3^m-2m-1,4]$ or
	$[3^m-1,3^m-2m-2,5]$.
	In 2016, Wang and Wu \cite{wang2016several} presented
	four classes of optimal
	ternary cyclic codes with parameters
	$[3^m-1,3^m-2m-1,4]$ by analyzing the solutions
	of certain equations over $\F_{3^m}$. It was shown that
	some previous results on  optimal ternary cyclic codes given in \cite{ding2013optimal}\cite{li2014optimal}\cite{ding2013five}\cite{zhou2013seven}
	are special cases of the constructions given in \cite{wang2016several}.
	In 2019, another open problem 7.12  proposed in \cite{ding2013optimal} was settled independently in      \cite{han2019open}  and \cite{liu2022some}.
	Later,  Zha et al.   \cite{zha2020new}\cite{zha2021further} presented several  classes of optimal ternary cyclic codes $\mathcal{C}_{(1,e)}$ and $\mathcal{C}_{(u,v)}$, they also proposed 
	a link between the ternary cyclic codes $\mathcal{C}_{(1,e)}$ and $\mathcal{C}_{(\frac{3^m+1}{2},e+\frac{3^m-1}{2})}$. In 2022, 
	Zhao, Luo, and Sun \cite{zhao2022two} presented
	two families of optimal ternary cyclic codes and solved the remain problem in \cite{wang2016several}.  
    Recently,  Ye and Liao \cite{ye2023ding}  gave a  counterexample of the open problem 7.13 proposed in \cite{ding2013optimal} and presented three classes of optimal ternary cyclic codes  $\mathcal{C}_{(1,e)}$. 
   A sufficient and necessary condition for   the ternary cyclic code
	$\mathcal{C}_{(u,v)}$ to be optimal  were given in \cite{wang2022several}. Based on this basic result,  several classes of optimal ternary cyclic codes $\mathcal{C}_{(u,v)}$ were also presented
	in \cite{wang2022several}. Recently, Li and Liu \cite{li2023some} proposed some classes of  optimal ternary cyclic codes $\mathcal{C}_{(2,e)}$ with parameters $[3^m-1,3^m-2m-1,4]$.
	There have been some other  optimal ternary cyclic codes constructed in the literature, see
	\cite{fan2016class}\cite{liu2021two}\cite{yan2018family}\cite{fan2022two} and the references therein.
	We list the known optimal ternary cyclic codes  $\mathcal{C}_{(1,e)}$ and  $\mathcal{C}_{(u,v)}$ with parameters  $[3^m-1,3^m-2m-1,4]$ in Tables  \ref{tab:1111} and \ref{tab:man},  respectively. 
	
	However, there are only a few constructions of optimal $p$-ary ($p\ge 3$ is a prime)
	cyclic codes $\mathcal{C}_{(0,1,e)}$ and $\mathcal{C}_{(1,e,s)}$ ($s=\frac{p^m-1}{2}$)  with parameters  $[p^m-1,p^m-\frac{3m}{2}-2,4]$. 
	In 2019, Li, Zhu, and Liu \cite{li2019three} 
	constructed a class of optimal ternary cyclic codes  $\mathcal{C}_{(0,1,e)}$ with parameters  $[3^m-1,3^m-\frac{3m}{2}-2,4]$, where  
	$e=3^{\frac{m}{2}}+1$. Recently, 
	Wu, Liu, and Li \cite{wu2023several} generalized the construction in \cite{li2019three} and  presented two classes of optimal $p$-ary (for any odd prime $p$) cyclic codes
	$\mathcal{C}_{(0,1,e)}$ and  $\mathcal{C}_{(1,e,s)}$ with parameters  $[p^m-1,p^m-\frac{3m}{2}-2,4]$. 
	We   summarize some known optimal  $p$-ary cyclic codes  $\mathcal{C}_{(0,1,e)}$ and  $\mathcal{C}_{(1,e,s)}$ with parameters  $[p^m-1,p^m-\frac{3m}{2}-2,4]$
	in Table \ref{tab:2222}.

	In this paper, 	we first  present four classes of optimal ternary cyclic codes $\mathcal{C}_{(0,1,e)}$
	and      $\mathcal{C}_{(1,e,s)}$  with parameters  $[3^m-1,3^m-\frac{3m}{2}-2,4]$  	by analyzing
	the solutions of certain equations  over finite fields,  where $s=\frac{3^m-1}{2}$.  
	Then  we give four  classes  of
	optimal ternary  cyclic codes   with parameters  $[3^m-1,3^m-2m-1,4]$ by
	analyzing the irreducible factors of certain polynomials and  determining the solutions of certain equations  over $\F_{3^m}$.  
     Two of them confirmed some special cases  of the open problem 7.9 proposed in  \cite{ding2013optimal}. 
	We also show that the new cyclic codes constructed in this paper are 
	inequivalent to the known ones.

	The rest of this paper is organized
	as follows. In Section \ref{Section2}, we introduce some preliminaries.
	Four  classes of optimal ternary cyclic codes $\mathcal{C}_{(0,1,e)}$ and  $\mathcal{C}_{(1,e,s)}$ are given in Section \ref{Section3}.
	In Section \ref{Section4}, we present four  classes   of
	optimal ternary cyclic codes with two zeros. Concluding  remarks  are given in Section \ref{Section5}.

\section{Preliminaries}\label{Section2}

Let $p$ be a prime and $m$ be a positive integer.
The $p$-cyclotomic coset modulo $p^m-1$ containing $j$ is defined as
\[C_j=\{j\cdot p^r\pmod{(p^m-1)}: r=0,1,\cdots,l_j-1\},\]
where $l_j$ is the least positive integer such that $j\cdot p^{l_j} \equiv j \bmod{(p^m-1)}.$
Thus the size of $C_j$ is $|C_j|=l_j$. It is known that $l_j|m$. The smallest integer
in $C_j$ is called the coset leader of $C_j$.

The following lemmas will be frequently used in the sequel.

\begin{lemma}\cite{xu2016optimal} \label{ce}
	Let $p$ be a prime and $m$ be a positive
	integer. Let $n={{p}^{m}}-1$.  For any $1\le e\le n-1$,  if
	one of the following conditions is satisfied, then $l_e=|{C_{e}} |=m$:
	\begin{itemize}
		\item [1)]  $1\le \gcd (e, n)\le p-1$,
		\item  [2)]  $ \gcd (e, n)\cdot \gcd(p^j-1,n) \notequiv 0\pmod n$
		for all $1\le j<m$.
	\end{itemize}
\end{lemma}

\begin{lemma}\cite{xu2016optimal}
	\label{wutingting}
	Let  $e=p^k+1$, where  $1\le k\le m-1$.  Then 
	\begin{enumerate}
		\item [1)] If $m$ is odd, then $|{C_{e}} |=m$;
		\item  [2)] If $m$ is even, then
		$$|{C_{e}} |=\begin{cases}
			\frac{m}{2},\ k=\frac{m}{2}\\
			m,\ k\neq \frac{m}{2}.
		\end{cases}$$
	\end{enumerate}
\end{lemma}

\begin{lemma}\cite{zha2020new}\label{le1}
	Let $m$ be odd and $e=\frac{3^{h}+5}{2}$, where $h$ is an odd integer. Then $e\notin C_{1}$ and $|C_{e}|=m$.
	\end{lemma}

By the following bound of the Hamming distance of general linear codes, it can be  shown that 
cyclic code    ${\mathcal{C}_{(1,e,s)}}$ or    ${\mathcal{C}_{(0,1,e)}}$ with parameters $[3^m-1, 3^m-\frac{3m}{2}-2,4]$ is optimal. 
\begin{lemma}\cite{el2007bounds}
	\label{spherepackingtight}
	Let $\mathcal{A}_{p}(n,d)$ be the maximum number of codewords of a p-ary code with length n and Hamming distance at least $d$. If $p\ge 3$, $t=n-d+1$ and $r=\left\lfloor min\left\{ \frac{n-t}{2}, \frac{t-1}{p-2}\right\}\right\rfloor$, then$$\mathcal{A}_{p}(n,d)\le \frac{p^{t+2r}}{\sum_{i=0}^{r}\tbinom{t+2r}{i}(p-1)^{i}}$$
\end{lemma}

Ding and Helleseth proved the following fundamental theorem about the optimality
of the ternary cyclic codes $\mathcal{C}_{(1,e)}$.

\begin{lemma}\label{thm-DH}{\cite{ding2013optimal}} 
Let $e\not\in C_1$ and $|C_e|=m$. The ternary cyclic code $\mathcal{C}_{(1,e)}$ has parameters $[3^m-1,3^m-1-2m,4]$ if and only if the following conditions are satisfied:
\begin{itemize}
     \item[1)] $e$ is even;
     \item[2)] $(x+1)^e+x^e+1=0$ has no solution in   $\mathbb{F}_{3^m}\setminus \F_3$; and
     \item[3)] $(x+1)^e-x^e-1=0$ has no solution  in $\mathbb{F}_{3^m}\setminus \F_3$.
\end{itemize}
\end{lemma}

\begin{lemma}\cite{wang2022several}\cite{li2023some}\label{lastdance} 
	Let $e$ be a positive integer with $1\le e\le 3^m-1$ and $|C_e|=m$. The ternary cyclic code $C_{(2,e)}$ has parameters $[3^m-1, 3^m-2m-1, 4]$ if and only if the following three conditions are satisfied:
	\begin{enumerate}
		\item  [1)] $e$ is odd;
		\item  [2)]  the equation $(1+x^2)^e-(1+x^e)^2=0$ has no solution  in $\mathbb{F}_{3^m}\setminus\F_3$; and 
		\item  [3)]  the equation $(1+x^2)^e+(1+x^e)^2=0$ has no solution  in $\mathbb{F}_{3^m}\setminus\F_3$.
	\end{enumerate}
\end{lemma}

\begin{longtable}{|c|c|c|c|}
	\caption{Known  optimal ternary cyclic codes $\mathcal{C}_{(1,v)}$ with parameters $[3^m-1,3^m-2m-1,4]$}
	\label{tab:1111}\\
	\hline
	Type & \begin{tabular}[c]{@{}c@{}}$v$ (even) \end{tabular} & Conditions & Ref. \\ \hline
	\endfirsthead
	%
	%
	1 & $(3^k+1)/2$ & $m\ge2$, $k$ is odd, $\gcd(m,k)=1$ & \cite{carlet2005linear} \\ \hline
	2 & $3^k+1$ & $m\ge2$, $m/\gcd(m,k)$ is odd & \cite{carlet2005linear} \\ \hline
	3 & $(3^m-3)/2$ & $m$ is odd, $m\ge5$ & \cite{ding2013optimal} \\ \hline
	4 & $(3^m+1)/4+(3^m-1)/2$ & $m$ is odd, $m\ge3$ & \cite{ding2013optimal} \\ \hline
	5 & $(3^{(m+1)/4}-1)(3^{(m+1)/2}+1)$ & $m\equiv3\ ({\rm{mod}}\ 4)$ & \cite{ding2013optimal} \\ \hline
	6 & $(3^{(m+1)/2}-1)/2$ or $(3^{m+1}-1)/8$ & $m\equiv3\ ({\rm{mod}}\ 4)$ & \cite{ding2013optimal} \\ \hline
	7 & $(3^{(m+1)/2}-1)/2+(3^m-1)/2$ & $m\equiv1\ ({\rm{mod}}\ 4)$ & \cite{ding2013optimal} \\ \hline
	8 & $(3^{m+1}-1)/8+(3^m-1)/2$ & $m\equiv1\ ({\rm{mod}}\ 4)$ & \cite{ding2013optimal} \\ \hline
	9 & $(3^h-1)/2$ & $m$ is odd, $h$ is even, $\gcd(m,h)=\gcd(m,h-1)=1$ & \cite{ding2013optimal} \\ \hline
	10 & $3^h-1$ & $\gcd(m,h)=\gcd(3^m-1,3^h-2)=1$ & \cite{ding2013optimal} \\ \hline
	11 & $2(3^{m-1}-1)$ or $5(3^{m-1}-1)$ or $16$ & $m$ is odd, $m\not\equiv0\ ({\rm{mod}}\ 3)$ & \cite{li2014optimal} \\ \hline
	12 & $(3^m-1)/2-2$, or $(3^m-1)/2+10$ & $m\equiv2\ ({\rm{mod}}\ 4)$ & \cite{li2014optimal} \\ \hline
	13 & $(3^m-1)/2-5$ or $(3^m-1)/2+7$ or $20$ & $m$ is odd & \cite{li2014optimal} \\ \hline
	14 & $2(3^h+1)$ & $m$ is odd & \cite{li2015conjecture} \\ \hline
	15 & $(3^m-3)/4$ & $m$ is odd & \cite{yan2019new} \\ \hline
	16 & $3^h+5$ & $m\equiv0\ ({\rm{mod}}\ 4)$, $m\ge4$, $h=m/2$ & \cite{han2019open} \\ \cline{3-3}
	&  & $m\equiv 2\ ({\rm{mod}}\ 4)$, $m\ge 6, h=(m+2)/2$ &  \\ \cline{3-3}
	&  & $m$ is odd, $\gcd(m,3)=1$, $2h\equiv \pm 1\ ({\rm{mod}}\ m)$ & \\ \cline{3-4}
	&  & $m\ge 5$ is prime,  $m\neq 19$,   $2h\equiv 3\ ({\rm{mod}}\ m) $  & \cite{ye2023ding} \\ \cline{3-3}
     &  & $m\ge 5$ is prime,  $m\equiv 2\ ({\rm{mod}}\ 3) $,  $3h\equiv 1\ ({\rm{mod}}\ $  $m) $   & \\ \hline
	17 & $v(3^s-1)\equiv\ 3^t-1({\rm{mod}}\ 3^m-1)$ & $\gcd(m,t)=\gcd(m,t-s)=1$ & \cite{wang2016several}\cite{zhao2022two} \\ \hline
	18 & $v(3^s+1)\equiv\ 3^t+1({\rm{mod}}\ 3^m-1)$ & $\gcd(m,t+s)=\gcd(m,t-s)=1$ & \cite{wang2016several}\cite{zhao2022two} \\ \hline
	19 & $v\equiv(3^m-1)/2+3^s+1\ ({\rm{mod}}\ 3^m-1)$ & $m$ is even, $m/\gcd(m,s)$ is odd & \cite{wang2016several} \\ \hline
	20 & $v\equiv(3^m-1)/2+3^s-1\ ({\rm{mod}}\ 3^m-1)$ & $m$ is even, $\gcd(m,s)=\gcd(3^m-1,3^s-2)=1$ & \cite{wang2016several}\cite{zhao2022two} \\ \hline
	21 & $(3^h+7)/2$ & $m$ is odd, $h$ is even, $1\le h\textless m$ & \cite{zha2020new} \\ \hline
	22 & $(3^h+7)/2+(3^m-1)/2$ & $m$ is odd, $h$ is odd, $1\le h\textless m$ & \cite{zha2020new} \\ \hline
	23 & $(3^{(m+1)/2}+5)/2$ & $m\equiv1\ ({\rm{mod}}\ 4)$, $m\not\equiv0\ ({\rm{mod}}\ 3)$ & \cite{zha2020new} \\ \hline
	24 & $3^h+13$ & $m$ is odd, $\gcd(m,3)=1$, $2h\equiv \pm 1\ ({\rm{mod}}\ m)$ & \cite{han2019open} \\ \hline
	25 & $(3^{(m+1)/2}+5)/2+(3^m-1)/2$ & $m\equiv3\ ({\rm{mod}}\ 4)$, $m\not\equiv0\ ({\rm{mod}}\ 3)$ & \cite{zha2020new} \\ \hline
	26 & $v(3^h+1)\equiv(3^m+1)/2\ ({\rm{mod}}\ 3^m-1)$ & $m$ is odd & \cite{zha2020new} \\ \hline
	27 & $5v\equiv2\ ({\rm{mod}}\ 3^m-1)$ & $m\not\equiv0\ ({\rm{mod}}\ 3)$ & \cite{zha2020new} \\ \hline
	28 & $7v\equiv2\ ({\rm{mod}}\ 3^m-1)$ & $m\not\equiv0\ ({\rm{mod}}\ 5)$, $\gcd(m,6)=1$ or $3$ & \cite{zha2020new} \\ \hline
	29 & $5v\equiv4\ ({\rm{mod}}\ 3^m-1)$ & $m\textgreater 2$, $m\not\equiv0\ ({\rm{mod}}\ 3)$, $m\not\equiv0\ ({\rm{mod}}\ 5)$ & \cite{zha2020new} \\ \hline
	30 & $5v\equiv 3^m-3\ ({\rm{mod}}\ 3^m-1)$ & $m\not\equiv0\ ({\rm{mod}}\ 3)$, $m\not\equiv0\ ({\rm{mod}}\ 4)$ & \cite{fan2022two} \\ \hline
	31 & $7v\equiv 3^m-3\ ({\rm{mod}}\ 3^m-1)$ & $m$ is odd, $m\not\equiv 0\ ({\rm{mod}}\ 3)$, $m\not\equiv 0\ ({\rm{mod}}\ 7)$ & \cite{fan2022two} \\ \hline
	32 & $5v\equiv 3^m-5\ ({\rm{mod}}\ 3^m-1)$ & $m$ is odd, $m\not\equiv 0\ ({\rm{mod}}\ 5)$ & \cite{fan2022two} \\ \hline
	33 & $(3^{(m-1)/2}+5)/2$ & $m \not\equiv 0 \pmod 3,$ $m\equiv 3\ ({\rm{mod}}\ 4)$ & \cite{ye2023ding} \\ \hline
	34 & $(3^{(m-1)/2}+5)/2+(3^m-1)/2$ & $m \not\equiv 0 \pmod 3,$ $m\equiv 1\ ({\rm{mod}}\ 4)$ & \cite{ye2023ding} \\ \hline
	35 & $(3^m-1)/2-k$ ($m>1$ is odd)  & $k=7$, or $k=11, -19$ and $m\not\equiv 0\ ({\rm{mod}}\ 9)$ & \cite{wang2022several} \\ \hline
	36 & $(3^m-1)/2+3^{s}+2$ ($m$ is odd)  & $\gcd(m,s)=1$, $x^{3^s+1}-x^2+1=0$ has no solution in  $\mathbb{F}_{3^m}$ & \cite{zhao2022two} \\ \hline
\end{longtable}

\begin{longtable}{|c|c|c|c|c|}
	\caption{Known optimal ternary cyclic codes $\mathcal{C}_{(u,v)}$  with parameters $[3^m-1,3^m-2m-1,4]$}
	\label{tab:man}\\
	\hline
	Type & \begin{tabular}[c]{@{}c@{}}$u$\end{tabular} & \begin{tabular}[c]{@{}c@{}}$v$\end{tabular} & Conditions & Ref. \\ \hline
	\endfirsthead
	\multicolumn{5}{c}%
	{{\bfseries Table \thetable\ continued from previous page}} \\
	\hline
	Type & \begin{tabular}[c]{@{}c@{}}Value of $u$ or \\ restrictions on $u$\end{tabular} & \begin{tabular}[c]{@{}c@{}}Value of $v$ or\\ restrictions on $v$\end{tabular} & Conditions & Ref. \\ \hline
	\endhead
	1 & $\frac{3^m+1}{2}$ & $(3^k+1)/2$ & $m$ is odd, $k$ is even, $\gcd(m,k)=1$ & \cite{zhou2014class} \\ \cline{3-5} 
	&  & $2\cdot 3^{(m-1)/2}+1$ & $m\ge 3$, $m$ is odd & \cite{fan2016class} \\ \cline{3-5} 
	&  & $3^s+2$ & \begin{tabular}[c]{@{}c@{}}$m$ is odd, $m\ge 3$, $4s\equiv 1\ ({\rm{mod}}\ m)$, $9\nmid m$\end{tabular} & \cite{yan2018family} \\ \cline{3-5} 
	&  & $(3^m-1)/2+e$ &\begin{tabular}[c]{@{}c@{}} $m$ is odd, $e$ is even, $\mathcal{C}_{(1,e)}$ has\\parameters $[3^m-1,3^m-2m-1,4]$ \end{tabular} & \cite{zha2021further} \\ \cline{3-5} 
	&  & $(3^{m+1}+7)/8$ & \begin{tabular}[c]{@{}c@{}}$m\ge 3$, $m\equiv 3\ ({\rm{mod}}\ 4)$, $9\nmid m$, $5\nmid m$\end{tabular} & \cite{qiu2023class} \\ \cline{3-5} 
	&  &  \begin{tabular}[c]{@{}c@{}} $3^h+2\cdot 3^i$, \\ $h>1$, \\  $m=2h-1$ \end{tabular}  & $m\not\equiv 0\ ({\rm{mod}}\ 3)$, $i=0$ & \cite{he2023two}  \\ \cline{4-4}
	&  &  & \begin{tabular}[c]{@{}c@{}}$m\not\equiv 0\ ({\rm{mod}}\ 3)$, $i=1$, $h\not\equiv 3\ ({\rm{mod}}\ 5)$\end{tabular} &  \\ \cline{4-4}
	&  &  & \begin{tabular}[c]{@{}c@{}}$m\not\equiv 0\ ({\rm{mod}}\ 3)$, $i=2$, $h\not\equiv 27\ ({\rm{mod}}\ 53)$\end{tabular} &  \\ \hline
	2 & $2^i$ & $(3^m-1)/2+2^i \cdot e$ & \begin{tabular}[c]{@{}c@{}}$m$ is odd, $e$ is even, and the ternary\\ cyclic code $\mathcal{C}_{(1,e)}$ is optimal\end{tabular} & \cite{zha2021further} \\ \hline
	3 & $u$ & $(3^k-1)u+(3^m-1)/2$ & \begin{tabular}[c]{@{}c@{}} $m$ is odd, $\gcd(u,3^m-1)=2$,\\  $k=1,2,3,(m+1)/2$ and $\gcd(m,k)=1$\end{tabular} & \cite{zha2021further} \\ \hline
	4 & $\frac{3^k+1}{2}$ & $(3^l+1)/2$ & \begin{tabular}[c]{@{}c@{}}$m$, $l$, and $\frac{m}{\gcd(m,l)}$  are all even,  \\  $\gcd(m,k+l)=\gcd(m,k-l)=1$\end{tabular} & \cite{zha2021further} \\ \hline
	5 &   \begin{tabular}[c]{@{}c@{}}  $3^m-6$, \\  $m$ is even \end{tabular}  & $(3^k+1)/2$, $k$ is odd &  \begin{tabular}[c]{@{}c@{}}$k=m-1$, $m\not\equiv 0({\rm{mod}}\ 6)$, $m\not\equiv 0({\rm{mod}}\ 20)$\end{tabular} & \cite{wang2022several} \\ \cline{4-4}
	&  &  & \begin{tabular}[c]{@{}c@{}}$k=1$, $m\not\equiv 0({\rm{mod}}\ 6)$, $m\not\equiv 0\ ({\rm{mod}}\ 20)$\end{tabular} &  \\ \cline{4-4}
	&  &  & \begin{tabular}[c]{@{}c@{}}$k=3$, $6\nmid m, 25\nmid m,$\\ $46\nmid m$, $78\nmid m$\end{tabular} &  \\ \hline
	6 & $2$ & $(3^m-1)/2+2(3^k-1)$ & \begin{tabular}[c]{@{}c@{}}$m,k$ are positive integers, $m$ is odd,\\ $\gcd(m,k)=\gcd(3^k-2,3^m-1)=1$\end{tabular} & \cite{liu2021two} \\ \cline{3-5} 
	&  & $(3^m-1)/2+2(3^k+1)$ & \begin{tabular}[c]{@{}c@{}}$m,k\in \mathbb{N}^{*}$, $m$ is odd, $\gcd(m,k)=1$\end{tabular} & \cite{liu2021two} \\ \cline{3-5} 
	&  & $(3^k+1)/2$ & \begin{tabular}[c]{@{}c@{}}$k$ is even, $2\le k\le m$\\ $\gcd(k+1,m)=\gcd(k-1,m)=1$\end{tabular} & \cite{fan2022two} \\ \cline{3-5}
	&  & $3v\equiv 5\ ({\rm{mod}}\ 3^m-1)$ & $m\not\equiv 0\ ({\rm{mod}}\ 3)$ & \cite{fan2022two} \\ \cline{3-5} 
	&  & $3v\equiv 7\ ({\rm{mod}}\ 3^m-1)$ & \begin{tabular}[c]{@{}c@{}}$m\not\equiv 0\ ({\rm{mod}}\ 5)$, $m\not\equiv 0\ ({\rm{mod}}\ 6)$\end{tabular} & \cite{fan2022two} \\ \cline{3-5} 
	&  & $3v\equiv 11\ ({\rm{mod}}\ 3^m-1)$ & \begin{tabular}[c]{@{}c@{}}$m\not\equiv 0\ ({\rm{mod}}\ 4)$, $m\not\equiv 0\ ({\rm{mod}}\ 9)$\end{tabular} & \cite{fan2022two} \\ \cline{3-5} 
	&  & $3v\equiv 13\ ({\rm{mod}}\ 3^m-1)$ & \begin{tabular}[c]{@{}c@{}}$m\not\equiv 0\ ({\rm{mod}}\ 3)$, $m\not\equiv 0\ ({\rm{mod}}\ 4)$\end{tabular} & \cite{fan2022two} \\ \cline{3-5} 
	&  & $(3^m+2\cdot 3^k+1)/2$ & $m$ is odd, $\gcd(m,k)=1$ & \cite{li2023some} \\ \cline{3-5} 
	&  & $(3^m-9)/2$ & $m$ is odd & \cite{li2023some} \\ \cline{3-5} 
	&  & $3^h+2$ & \begin{tabular}[c]{@{}c@{}}$h\in \mathbb{N}^{*}$, $2h\equiv 1\ ({\rm{mod}}\ m)$\end{tabular} & \cite{li2023some} \\ \cline{3-5} 
	&  & $\frac{3^m-1}{2}+2(2\cdot 3^k-1)$ & $m$ is odd, $k=(m-1)/2$ & \cite{li2023some} \\  \hline
\end{longtable}

	\begin{longtable}{|c|c|c|c|}
		\caption{Known optimal cyclic codes with parameters $[p^m-1, p^m-\frac{3m}{2}-2, 4]$}
		\label{tab:2222}\\
		\hline
		Type  & $e$ & Conditions & Ref. \\ \hline
		\endfirsthead
		\multicolumn{4}{c}%
		{{\bfseries Table \thetable\ continued from previous page}} \\
		\hline
		Type of cyclic code & Value of $e$ or restriction on $e$ & Conditions & Ref. \\ \hline
		\endhead
		$\mathcal{C}_{(1,e,s)}$ & $e=\frac{p^m-1}{2}+1+p^{\frac{m}{2}}$ & $p$ is an odd prime & \cite{wu2023several} \\ \hline
		$\mathcal{C}_{(0,1,e)}$ & $e=1+p^{\frac{m}{2}}$ & $p$ is an odd  prime & \cite{wu2023several} \\ \hline 
	\end{longtable}

\begin{lemma}\cite{lidl1983encyclopedia}\label{factorlem}
	For every finite field $\mathbb{F}_{p^m}$ and every positive integer $r$,
	the product of all monic  irreducible polynomials over $\mathbb{F}_{p^m}$ whose degrees divide $r$ is equal to ${x}^{(p^m)^{r}}-x$.
\end{lemma}

\begin{lemma}\cite{lidl1983encyclopedia}   \label{rootlem}
	Let  $f(x)$ be an irreducible
	polynomial over $\F_{p^m}$ of degree $r$.
	Then $f(x) = 0$ has a root $x$ in $\F_{p^{mr}}$. Furthermore, all the
	roots of $f(x) = 0$ are simple and are given by the $r$  distinct elements $x, x^{p^m}, x^{p^{2m}}, \cdots ,
	x^{p^{m(r-1)}}$ of $\F_{p^{mr}}$.
\end{lemma}

\section{Four   classes of  optimal ternary cyclic codes with parameters $[3^m-1, 3^m-\frac{3m}{2}-2, 4]$}\label{Section3}
In this section, 
we will
propose four classes of optimal ternary   cyclic codes ${\mathcal{C}_{(0,1,e)}}$ or  ${\mathcal{C}_{(1,e,s)}}$  with parameters $[3^m-1, 3^m-\frac{3m}{2}-2, 4]$
by analyzing the solutions of certain equations over $\F_{3^m}$.

\subsection{The first  two classes  of  optimal ternary cyclic codes with parameters $[3^m-1, 3^m-\frac{3m}{2}-2, 4]$}

Let $e=2\cdot 3^{m-1}-3^{\frac{m}{2}-1}-1$, where $m$ is an even integer. 
In this subsection, we will show that
${\mathcal{C}_{(0,1,e)}}$ is an  optimal ternary cyclic code with parameters  $[3^m-1, 3^m-\frac{3m}{2}-2,4]$.

\begin{theorem}\label{p3-thm}
	Let  $m$ be an  even integer and  $e=2\cdot 3^{m-1}-3^{\frac{m}{2}-1}-1$.  Then $\mathcal{C}_{(0,1,e)}$ is an optimal ternary  cyclic code with parameters $[3^m-1, 3^m-\frac{3m}{2}-2,4]$.
\end{theorem}

{\em Proof:}
Note that $3e\equiv -1-3^{\frac{m}{2}}\pmod {3^m-1}$. It is obviously that 
$e\notin {{C}_{1}}$ and $|C_e|=\frac{m}{2}$.
Thus  the dimension of ${\mathcal{C}_{(0,1,e)}}$ is ${{3}^{m}}-\frac{3m}{2}-2$.

In the following we show that
${\mathcal{C}_{(0,1,e)}}$ does not have a nonzero  codeword of Hamming weight less than 4.
Clearly,  the minimum distance of $\mathcal{C}_{(0,1,e)}$ cannot  be 1.
Suppose that   ${\mathcal{C}_{(0,1,e)}}$ has a codeword of Hamming weight 2, then  there exist two elements ${{c}_{1}}$,
${{c}_{2}}\in {{\mathbb{F}}_{3}^*}$ and two distinct elements ${{x}_{1}} $,
${{x}_{2}}\in {{\mathbb{F}}_{3^m}^*}$ such that
\begin{equation}\label{weight2equ}
	\left\{ \begin{array}{{l}}
		c_1+c_2=\ 0\\
		c_1x_1+c_2x_2=\ 0\\
		c_1{x_1^e}+c_2{x_2^e}=\ 0.
	\end{array} \right.
\end{equation}
By the first equation of (\ref{weight2equ}), 
we have ${c}_{1}=-{c}_{2}$, thus  ${x}_{1}={x}_{2}$ by the second equation of  (\ref{weight2equ}), which is contrary to ${x}_{1}\ne {x}_{2}$.
Therefore,  $\mathcal{C}_{(0,1,e)}$ does not have  a codeword of Hamming weight 2.

$\mathcal{C}_{(0,1,e)}$ has a codeword of Hamming weight 3 if and only if there exist three elements ${c}_{1},{c}_{2},{c}_{3}\in \mathbb{F}_{3}^*$ and three distinct elements ${x}_{1},{x}_{2,},{x}_{3}\in \mathbb{F}_{3^m}^*$ such that
\begin{equation}\label{pm-21}
	\left\{ \begin{array}{l}
		{c}_{1}+{c}_{2}+{c}_{3}=0.\\
		{c}_{1}{x}_{1}+{c}_{2}{x}_{2}+{c}_{3}{x}_{3}=0  \\
		{c}_{1}x_{1}^{e}+{c}_{2}x_{2}^{e}+{c}_{3}x_{3}^{e}=0.		
	\end{array} \right.
\end{equation}
By the first equation of (\ref{pm-21}), we have $(c_1,c_2,c_3)=(1,1,1)$ or $(c_1,c_2,c_3)=(-1,-1,-1)$.   
Due to symmetry, it is sufficient to  consider the case  $(c_1,c_2,c_3)=(1,1,1)$. In this case, let $x=\frac{x_1}{x_3}$
and  $y=\frac{x_2}{x_3}$, then $x\ne 1, $   $y\ne 1 $, and  (\ref{pm-21}) becomes 

\begin{equation}\label{pm-22}
	\left\{ \begin{array}{l}
		x+y+1=0\\
		x^e+y^e+1=0, 
	\end{array} \right.
\end{equation}
which is equivalent to $(x+1)^e+x^e+1=0$. 
Raising both sides of this equation to the power of 3 will lead to 
$(x+1)^{3e}+x^{3e}+1=0$, i.e.,  
$$x^{2 (3^{\frac{m}{2}}+1)}+x^{2\cdot 3^{\frac{m}{2}}+1}+x^{3^{\frac{m}{2}}+2}+x^{3^{\frac{m}{2}}}+x+1=0.$$
Notice that the above equation can be factorized as 
 $$\big((x^{3^{\frac{m}{2}}}-1)(x-1)-(x^{3^{\frac{m}{2}}}-x)\big)\big((x^{3^{\frac{m}{2}}}-1)(x-1)+(x^{3^{\frac{m}{2}}}-x)\big)=0.$$
As a result, $(x^{3^{\frac{m}{2}}}-1)(x-1)-(x^{3^{\frac{m}{2}}}-x)=0$ or $(x^{3^{\frac{m}{2}}}-1)(x-1)+(x^{3^{\frac{m}{2}}}-x)=0$.

Case 1),
$(x^{3^{\frac{m}{2}}}-1)(x-1)-(x^{3^{\frac{m}{2}}}-x)=0$.  
In this case, 
\begin{equation}\label{firsteq1}
	x^{3^{\frac{m}{2}}+1}+x^{3^{\frac{m}{2}}}+1=0.
\end{equation}
Taking 
$3^{\frac{m}{2}}$-th power of both sides of the above equation will lead to 
$x^{3^{\frac{m}{2}}+1}+x+1=0$, togehter with (\ref{firsteq1}), we have 
$x^{3^{\frac{m}{2}}}=x$, i.e.,  $x\in\F_{3^\frac{m}{2}}$. Thus, 
(\ref{firsteq1}) can be simplied to $x^2+x+1=0$, i.e.,  $x=1$, 
which is contrary to $x\ne 1$.  

Case 2),
$(x^{3^{\frac{m}{2}}}-1)(x-1)+(x^{3^{\frac{m}{2}}}-x)=0$.  
In this case,  $x^{3^{\frac{m}{2}}+1}+x+1=0.$ Similar to case 1), we can show that 
$x=1$,  which is contrary to $x\notin \{0,1\}$.

To sum up,   $\mathcal{C}_{(0,1,e)}$ does not have  a codeword of  Hamming weight 3.
This completes the proof.
\done

\begin{example}
	Let $p=3$ and  $m=4$.  Then  $e=2\cdot 3^{m-1}-3^{\frac{m}{2}-1}-1=50$.  Let $\alpha $ be the generator  of   $\mathbb{F}_{{3}^{4}}^*$
	with
	${\alpha}^{4}+2 {\alpha}^{3}+2=0$. Then the  code $\mathcal{C}_{(0,1,e)}$
	has  parameters $[80,73,4]$ and
	generator polynomial  $x^7+2x^6+x^5+x^3+2x+2$.
\end{example}

Similar as the proof of Theorem \ref{p3-thm}, we can obtain another class of  optimal ternary  cyclic codes  with parameters $[3^m-1, 3^m-\frac{3m}{2}-2,4]$. 
\begin{corollary}\label{coro-1}
	Let  $m$ be an  even integer and  $e= \frac{3^{m}-1}{2}-3^{\frac{m}{2}}-1$. Let $s=\frac{3^{m}-1}{2}$.
	Then $\mathcal{C}_{(1,e,s)}$ is an optimal ternary  cyclic code with parameters $[3^m-1, 3^m-\frac{3m}{2}-2,4]$.
\end{corollary}

\begin{example}
	Let $p=3$ and  $m=6$.  Then  $e= \frac{3^{m}-1}{2}-3^{\frac{m}{2}}-1=336$.  Let $\alpha $ be the generator  of   $\mathbb{F}_{{3}^{6}}^*$
	with 
	$\alpha^6+2\alpha^4+\alpha^2+2\alpha+2=0$. Then the  code $\mathcal{C}_{(1,e,s)}$
	has  parameters $[728,718,4]$ and
	generator polynomial $x^{10} + 2x^9 + 2x^6 + 2x^5 + 2x^4 + 2x^3 + 2x^2 + 2x + 1$.
\end{example}

\begin{remark}
It can be shown  that  the cyclic code  $\mathcal{C}_{(0,1,e)}$ constructed  in Theorem \ref{p3-thm} are not covered by the known ones in  Table \ref{tab:2222}.
 In fact, suppose that  $e=2\cdot 3^{m-1}-3^{\frac{m}{2}-1}-1$ and $3^{\frac{m}{2}}+1$
are in the same cyclotomic coset. 
Note that  $3e\equiv  2\cdot 3^{m}-3^{\frac{m}{2}}-3\equiv -3^{\frac{m}{2}}-1 \pmod{3^m-1}$. Thus, 
   there exists  an integer $1\le i\le m-1 $ such that $3e\equiv -3^{\frac{m}{2}}-1\equiv (3^{\frac{m}{2}}+1)\cdot 3^i\pmod{3^m-1}$,  i.e.,  
$(3^{\frac{m}{2}}+1)(3^i+1)\equiv 0 \pmod {3^m-1}$, which implies  $(3^{\frac{m}{2}}-1)|(3^i+1)$.  
It is known that 
		$$\gcd(3^{\frac{m}{2}}-1,3^i+1)=\begin{cases}
			2,\ {\rm if\, } \frac{m/2}{\gcd(i,m/2)} {\rm\,  is \, odd}\\
			3^{\gcd(i,m/2)}+1,\ {\rm if\,} \frac{m/2}{\gcd(i,m/2)} {\rm\, is \, even}.
		\end{cases}$$

If $m\equiv 2 \pmod 4,$ then $\gcd(3^{\frac{m}{2}}-1,3^i+1)=2$, i.e.,   $(3^{\frac{m}{2}}-1)\nmid (3^i+1)$. 
If $m\equiv 0 \pmod 4,$ then 
we have $ \gcd(3^{\frac{m}{2}}-1,3^i+1) \le  3^{\frac{m}{4}}+1$, one get $(3^{\frac{m}{2}}-1)\nmid (3^i+1)$ again. 
As a result, the cyclic code  $\mathcal{C}_{(0,1,e)}$  given in Theorem \ref{p3-thm}
is inequivalent to the known ones. Simiarly, we can show that the optimal code  $\mathcal{C}_{(1,e,s)}$ presented in Corollary \ref{coro-1} is  not covered by the known ones in  Table \ref{tab:2222}.

\end{remark}

\subsection{The second two  classes  of  optimal ternary cyclic codes with parameters $[3^m-1, 3^m-\frac{3m}{2}-2, 4]$}


\begin{theorem}\label{p3-thm2}
	Let $m\ge 4$ be an integer with $m\equiv 0\pmod 4$. Let   $e=\frac{3^m-1}{2}+3^{\frac{m}{2}}+1$.  Then $\mathcal{C}_{(0,1,e)}$ is an optimal ternary  cyclic code with parameters $[3^m-1, 3^m-\frac{3m}{2}-2,4]$.
\end{theorem}

{\em Proof:}
By Lemma \ref{wutingting},  we know that  $|{C}_{3^{\frac{m}{2}}+1}|=\frac{m}{2}$. Since  $(3^m-1)|e(3^l-1)$ if and only if  $(3^m-1)|(e+\frac{3^m-1}{2})(3^l-1)$,  
we have   $|{C}_{e}|=|{C}_{3^{\frac{m}{2}}+1}|=\frac{m}{2}$.
It is clearly that $\mathcal{C}_{(0,1,e)}$ does not have a  codeword of Hamming weight 1.
In the following we show that
$\mathcal{C}_{(0,1,e)}$ does not have a  codeword of Hamming weight 2 or 3.

Suppose that
$\mathcal{C}_{(0,1,e)}$ has a codeword of Hamming weight 2. Then there exist two elements ${{c}_{1}}$,
${{c}_{2}}\in {{\mathbb{F}}_{3}^*}$ and two distinct elements ${{x}_{1}} $,
${{x}_{2}}\in {{\mathbb{F}}_{p^m}^*}$ such that (\ref{weight2equ}) is satisfied.
By   (\ref{weight2equ}),  ${c}_{1}=-{c}_{2}$ and
${x}_{1}={x}_{2}$, which is contrary to ${x}_{1}\ne {x} _{2}$.
Thus $\mathcal{C}_{(0,1,e)}$ does not have a  codeword of Hamming weight 2.

$\mathcal{C}_{(0,1,e)}$ has  a codeword of  Hamming weight 3 if and only if   (\ref{pm-21}) has no pairwise different nonzero  solutions ${x}_{1},{x}_{2},{x}_{3}\in \F_{3^m}^*$.
Similar as  the proof of Theorem \ref{p3-thm}, we only need to show that 
$(x+1)^e+x^e+1=0$ has no solution in $\F_{3^m}\setminus \F_3$. 
Let $s=\frac{3^m-1}{2}$ and $k=\frac{m}{2}$.  Then $x^s=\pm 1$ and $(x+1)^s=\pm 1$. We consider the following four cases according to the values of $x^s$ and $(x+1)^s$.

Case 1), 	  	$({x}^{s},{(x+1)}^{s})=(1,1)$.
In this case,
$(x+1)^e+x^e+1=(x+1)^{3^k+1}+x^{3^k+1}+1=0$, i.e.,   $x^{3^k+1}-x^{3^k}-x+1=(x-1)^{3^k+1}=0$. Thus, $x=1$.

Case 2), 	  	$({x}^{s},{(x+1)}^{s})=(1,-1)$.
In this case,
$(x+1)^e+x^e+1=-(x+1)^{3^k+1}+x^{3^k+1}+1=0$, i.e.,   $x^{3^k}+x=x(x^{3^k-1}+1)=0$. Since 
$x\ne 0$, then $x^{3^k-1}=-1$. Therefore,  $$x^s=(x^{3^k-1})^{\frac{3^k+1}{2}}=-1$$ since $ \frac{3^k+1}{2}$ is odd due to 
$2|k$.  This  is contrary to $x^s=1$.

Case 3), 	  	$({x}^{s},{(x+1)}^{s})=(-1,-1)$.
In this case,
$(x+1)^e+x^e+1=-(x+1)^{3^k+1}-x^{3^k+1}+1=0$, i.e.,   $-x^{3^k}+x^{3^k-1}+1=0$. 
Thus, $(\frac{1}{x})^{3^k}+\frac{1}{x}-1=0$,   i.e.,  $$(\frac{1}{x}+1)^{3^k}+ \frac{1}{x}+1=0.$$
Since $x\ne -1$  due to ${(x+1)}^{s}=-1$, we have $(\frac{1}{x}+1)^{3^k-1}=-1$. Therefore, 
$$(\frac{x+1}{x})^{s}=(\frac{1}{x}+1)^{s}=((\frac{1}{x}+1)^{3^k-1})^{\frac{3^k+1}{2}}=-1$$ due to $ \frac{3^k+1}{2}$ is odd. 
This  is contrary to $(\frac{x+1}{x})^{s}=\frac{(x+1)^s}{x^s}=\frac{-1}{-1}=1$.

Case 4), 	  	$({x}^{s},{(x+1)}^{s})=(-1,1)$.
In this case,
$(x+1)^e+x^e+1=(x+1)^{3^k+1}-x^{3^k+1}+1=0$, i.e.,   ${x}^{3^k}+x-1=0$.
Thus,  $(x+1)^{3^k}+x+1=0 $. 
Since $x\ne -1$  due to ${(x+1)}^{s}=1$, we have $(x+1)^{3^k-1}=-1$. Therefore, 
$$(x+1)^{s}=((x+1)^{3^k-1})^{\frac{3^k+1}{2}}=-1$$ due to $ \frac{3^k+1}{2}$ is odd. 
This  is contrary to $(x+1)^{s}=1$.

As a consequence,
$\mathcal{C}_{(0,1,e)}$ does not have  a codeword of Hamming weight 3. This completes the proof. 
\done

\begin{example}
	Let $p=3$ and  $m=8$.  Then  $e= \frac{3^{m}-1}{2}+3^{\frac{m}{2}}+1=3362$.  Let $\alpha $ be the generator  of   $\mathbb{F}_{{3}^{8}}^*$
	with 
	$\alpha^8-\alpha^5+\alpha^4-\alpha^2-\alpha-1=0$. Then the  code $\mathcal{C}_{(0,1,e)}$
	has  parameters $[6560,  6547,   4]$ and
	generator polynomial $x^{13} + 2x^{11} + 2x^{10} + 2x^8 + x^7 + x^5 + 2x^4 + 2x^3 + 2$.
\end{example}

Similar as the proof of Theorem \ref{p3-thm2}, we can obtain another class of  optimal ternary  cyclic codes  with parameters $[3^m-1, 3^m-\frac{3m}{2}-2,4]$.

\begin{corollary}\label{coro-2}
	Let  $m\equiv 0\pmod 4$ and  $e= 3^{\frac{m}{2}}+1$. Let $s=\frac{3^{m}-1}{2}$.
	Then $\mathcal{C}_{(1,e,s)}$ is an optimal ternary  cyclic code with parameters $[3^m-1, 3^m-\frac{3m}{2}-2,4]$.
\end{corollary}

\begin{example}
	Let $p=3$ and  $m=8$.  Then  $e= 3^{\frac{m}{2}}+1=82$.  Let $\alpha $ be the generator  of   $\mathbb{F}_{{3}^{8}}^*$
	with 
	$\alpha^8-\alpha^5+\alpha^4-\alpha^2-\alpha-1=0$. Then the  code $\mathcal{C}_{(1,e,s)}$
	has  parameters $[6560,  6547,   4]$ and
	generator polynomial $x^{13} + 2x^{11} + 2x^{10} +  x^7 + 2x^3 + 2x^2 + 2x + 1$.

\end{example}

\begin{remark} It can be shown that  $e=\frac{3^m-1}{2}+3^{\frac{m}{2}}+1$  and $3^{\frac{m}{2}}+1$
are not in the same cyclotomic coset.
Actually, if  there exists  an integer $1\le i\le m-1 $ such that 
$3^i(3^{\frac{m}{2}}+1)\equiv \frac{3^m-1}{2}+3^{\frac{m}{2}}+1 \pmod {3^m-1}$, then   $\frac{3^{m/2}-1}{2}|(3^i-1)$. 
If $i\ne \frac{m}{2}$, then 
 $$\gcd(\frac{3^{\frac{m}{2}}-1}{2},3^i-1)\le \gcd(3^{\frac{m}{2}}-1,3^i-1)=3^{\gcd(m/2,i)}-1\le 3^{\frac{m}{4}}-1<\frac{3^{\frac{m}{2}}-1}{2}$$
due to $m\ge 4$. Therefore,  $\frac{3^{m/2}-1}{2}|(3^i-1)$ implies $i=\frac{m}{2}$. However, 
$3^{\frac{m}{2}}(3^{\frac{m}{2}}+1)\notequiv  \frac{3^m-1}{2}+3^{\frac{m}{2}}+1 \pmod {3^m-1}$.
  As a result, the cyclic code $\mathcal{C}_{(0,1,e)}$   given in Theorem \ref{p3-thm2}
is inequivalent to the known ones.  
 Simiarly, we can show that the optimal code $\mathcal{C}_{(1,e,s)}$ presented in Corollary \ref{coro-2} is  not covered by the known ones in  Table \ref{tab:2222}.

\end{remark}

\section{Optimal ternary cyclic codes with parameters $[3^m-1, 3^m-2m-1, 4]$}\label{Section4}

In this section, we will present a class  of optimal ternary cyclic codes ${\mathcal{C}_{(2,e)}}$ and 
three  classes of  optimal ternary cyclic codes ${\mathcal{C}_{(1,e)}}$. All of them has parameters $[3^m-1, 3^m-2m-1, 4]$.

\subsection{A class of optimal ternary cyclic codes ${\mathcal{C}_{(2,e)}}$ }
In this subsection, we will
consider the exponents $e=3^{\frac{m}{2}}+2$. We will show that 
${\mathcal{C}_{(2,e)}}$ is an optimal ternary cyclic codes with parameters 
$[3^m-1, 3^m-2m-1, 4]$ if $m\equiv 2 \pmod 4$.

\begin{theorem}\label{thm-3}
	Let  $e=3^{h}+2$, where $h=\frac{m}{2}$ and  $m\equiv 2 \pmod 4$. 
	Then $\mathcal{C}_{(2,e)}$ is an optimal cyclic code with parameters $[3^m-1, 3^m-2m-1,4]$.
\end{theorem}

{\em Proof:} Since $e$ is odd, we have  $e\notin {C}_{2}$. 
Note that $\gcd(e(3^{\frac{m}{2}}-2),3^m-1)=\gcd(3^m-4,3^m-1)=1$, which implies
$\gcd(e,3^m-1)=1$. 
According to Lemma \ref{ce}, we have  $|C_{e}|=m$.  
By Lemma \ref{lastdance}, we need to show that conditions 2) and 3) are met. 

We first consider the condition 2) in Lemma \ref{lastdance}, i.e.,  we show that  the equation 
$(x^2+1)^e-(x^e+1)^2=0$ has no solutions in $\F_{3^m}\setminus\F_3$.  Note that  $$(x^2+1)^e-(x^e+1)^2=x^{2(3^h+1)}-x^{2\cdot 3^h}-x^{3^h+2}-x^4+x^2=0.$$
The above equation can be factorized as  $$ x^2(x^{3^h}-1-(x^{3^h-1}+x))(x^{3^h}-1+x^{3^h-1}+x)=0.$$
Thus,  $x^{3^h}-1-(x^{3^h-1}+x)=0$ or $x^{3^h}-1+x^{3^h-1}+x=0$. 
We will show that both $x^{3^h}-1-(x^{3^h-1}+x)=0$ and  $x^{3^h}-1+x^{3^h-1}+x=0$ has no solution in $\F_{3^m}\setminus\F_3$.

Case 1), $x^{3^h}-1-(x^{3^h-1}+x)=0$. Note that $x=0$ and $x=1$ are  not  solutions of $x^{3^h}-1-(x^{3^h-1}+x)=0$. Thus, by $x^{3^h}-1-(x^{3^h-1}+x)=0$,  we have $x^{3^h}(1-\frac{1}{x})=x+1$, i.e.,  
\begin{equation}\label{neweq1}
	x^{3^h}=\frac{x^2+x}{x-1}.
\end{equation} 
Raising both sides of (\ref{neweq1}) to the power of $3^h$   gives  
\begin{equation}\label{neweq111}
 x^{3^{2h}}=\frac{{(x^{2\cdot 3^h})}+x^{3^h}}{x^{3^h}-1}.
\end{equation}
Note that  $ x^{3^{2h}}=x^{3^m}=x$. By equation (\ref{neweq1}), we have 
\begin{equation}\label{neweq11111}
x=\frac{x^4 + x^2 + 2x}{x^3 + 2x^2 + x + 2}. 
\end{equation}

If  $x^3 + 2x^2 + x + 2=(x+2)(x^2+1)=0 $, then $x=1$ or $x^2=-1$. Since $x\ne 1$, we have 
 $x^2=-1$.   By (\ref{neweq1}), $ x^{3^h}=1$, which implies $x=1$. This is contrary to 
$x^2=-1$. As a result, $x^3 + 2x^2 + x + 2\ne  0 $. Then by (\ref{neweq11111}), one obtains 
 $2x^3=0$, which implies $x=0$.

Case 2), $x^{3^h}-1+x^{3^h-1}+x=0$. Note that $x=0$ and $x=-1$ are  not  solutions  of $x^{3^h}-1+x^{3^h-1}+x=0$. Thus, by $x^{3^h}-1+x^{3^h-1}+x=0$,  we have $x^{3^h}(1+\frac{1}{x})=1-x$, i.e.,  
\begin{equation}\label{neweq2}
	x^{3^h}=\frac{-x^2+x}{x+1}.
\end{equation} 
Raising both sides of (\ref{neweq2}) to the power of $3^h$   will lead to $x^{3^{2h}}= x=\frac{-{(x^{2\cdot 3^h})}+x^{3^h}}{x^{3^h}+1}$.   By   (\ref{neweq2}), we have
\begin{equation}\label{neweq222}
x=\frac{-x^4 + x^3 -x^2 + x}{-x^3 + x^2 + 1}. 
\end{equation}
Similar as in Case 1), one can show that $-x^3 + x^2 + 1\ne 0$. Thus, according to (\ref{neweq222}), we have 
  $x^2=0$, which implies $x=0.$

Now  we consider the condition 3) in Lemma \ref{lastdance}, i.e.,  the  equation 
$(x^2+1)^e+(x^e+1)^2=0$  has no solution in $\F_{3^m}\setminus\F_3$. Note that 
 $$(x^2+1)^e+(x^e+1)^2=-(x^{2\cdot (3^h+2)}+x^{2\cdot (3^h+1)}-x^{2\cdot 3^h}+x^{3^h+2}-x^4+x^2+1), $$ and 
\begin{eqnarray}\label{neweq3}
	&& x^{2\cdot (3^h+2)}+x^{2\cdot (3^h+1)}-x^{2\cdot 3^h}+x^{3^h+2}-x^4+x^2+1\nonumber\\
	&=&x^{2\cdot 3^h}(x^4+x^2-1)+x^{3^h+2}-(x^4+x^2-1)-x^2\nonumber\\
	&=&(x^2-1)^{3^h}(x^4+x^2-1)+x^2(x^{3^h}-1)\nonumber\\
	&=&(x^{3^h}-1)((x^{3^h}+1)(x^4+x^2-1)+x^2)=0.
\end{eqnarray}

Therefore, we only need to show that 
\begin{equation}\label{neweq4}
	(x^{3^h}+1)(x^4+x^2-1)+x^2=0
\end{equation}
has no solution in $\F_{3^m}\setminus\F_3$.

If $x^4+x^2-1=0$, then from (\ref{neweq4}), we have $x^2=0$, i.e.,  $x=0$. Thus, $x^4+x^2-1\ne 0$. Therefore, 
from (\ref{neweq4}), we have 
\begin{equation}\label{neweq5}
	x^{3^h}=\frac{-x^4+x^2+1}{x^4+x^2-1}
\end{equation}
Thus, 
\begin{equation}\label{neweq52}
x^{2\cdot 3^h}=\frac{x^8+x^6-x^4-x^2+1}{x^8-x^6-x^4+x^2+1}=\frac{y_1}{y_2},
\end{equation}
where $y_1=x^8+x^6-x^4-x^2+1$   and  $y_2=x^8-x^6-x^4+x^2+1$. 
Similarly, one can obtain 
\begin{equation}\label{neweq53}
 x^{4\cdot 3^h}=\frac{y_1^2}{y_2^2}.
\end{equation}
Raising both sides of (\ref{neweq5}) to the power of $3^h$ gives 
$x^{3^{2h}}=\frac{-x^{4\cdot 3^h}+x^{2\cdot 3^h}+1}{x^{4\cdot 3^h}+x^{2\cdot 3^h}-1}$. Together with 
$ x^{3^{2h}}=x$, we have $$x=\frac{-x^{4\cdot 3^h}+x^{2\cdot 3^h}+1}{x^{4\cdot 3^h}+x^{2\cdot 3^h}-1}.$$
By (\ref{neweq52}) and (\ref{neweq53}), we have 
\begin{equation}\label{neweq533}
x((\frac{y_1}{y_2})^2+\frac{y_1}{y_2}-1)=-(\frac{y_1}{y_2})^2+\frac{y_1}{y_2}+1,
\end{equation}
 i.e., ,
\begin{equation}\label{neweq5333}
  x(y_1^2+y_1y_2-y_2^2)+y_1^2-y_1y_2-y_2^2=0.   
\end{equation}
By $y_1=x^8+x^6-x^4-x^2+1$   and  $y_2=x^8-x^6-x^4+x^2+1$, (\ref{neweq5333}) becomes  
$$f(x)\triangleq x^{17}-x^{16}+x^{15}+x^{14}+x^{11}+x^{10}-x^9+x^8-x^7-x^6-x^3-x^2+x-1=0.$$ 

Thanks to the Magma computation,  the  canonical factorization of $f(x)$ over $\F_3$  is given by 
$$f(x)=(x-1)^5(x^4+x-1)(x^4-x^3-1)(x^4-x^3+x^2-x+1).$$
Then by Lemma \ref{rootlem}, $f(x)=0$ has no solutions in $\F_{3^m}\setminus \F_3$ if and only if 
$m \notequiv 0 \pmod 4$.

Therefore,  $\mathcal{C}_{(2,e)}$ does not have a codeword of Hamming weight 3.
This completes the proof. \done

\begin{example}
	Let $p=3$ and  $m=6$.  Then  $e= 3^{\frac{m}{2}}+2=29$.  Let $\alpha $ be the generator  of   $\mathbb{F}_{{3}^{6}}^*$
	with 
	$\alpha^6+2\alpha^4+\alpha^2+2\alpha+2=0$. Then the  code $\mathcal{C}_{(2,e)}$
	has  parameters $[728,716,4]$ and
	generator polynomial $x^{12} -x^{11} + x^{10} -x^6 -x^3 -1$.
\end{example}

\begin{remark}
Note that $\gcd(3^{\frac{m}{2}}+2, 3^m-1)=1$. Thus,  the  code $\mathcal{C}_{(2,e)}$ given in Theorem \ref{thm-3} is equivalent to 
$\mathcal{C}_{(1,2e^{-1})}=\mathcal{C}_{(1,4\cdot 3^{\frac{m}{2}}-2)}$. Let $m=6$, we have  $4\cdot 3^{\frac{m}{2}}-2=106$. 
Magma
experiments confirm  that $106$ and all the exponents given in Table \ref{tab:1111} are not in the same coset. 
Thus, the optimal ternary cyclic codes $\mathcal{C}_{(1,2e^{-1})}$ given in Theorem \ref{thm-3}
are not covered by the known ones in Table \ref{tab:1111}. Take $m=6$, then $e=3^3+2=29$. 
Magma
experiments confirm  that $29$ and all the exponents given in  type 6 in  Table \ref{tab:man} are not in the same coset. 
Thus, the optimal ternary cyclic codes $\mathcal{C}_{(2,e)}$ given in Theorem \ref{thm-3}
are not covered by the known ones in Table \ref{tab:man}.
As a result, the optimal ternary cyclic code $\mathcal{C}_{(2,e)}$ given in Theorem \ref{thm-3} is inequivalent to the known ones in Tables \ref{tab:1111} and  \ref{tab:man}. 

\end{remark}

\subsection{Three  classes  of optimal ternary cyclic codes ${\mathcal{C}_{(1,e)}}$ }


\begin{theorem}\label{thm3h-1}
	Let $m$ be odd and $e$  be an even integer satisfying $e(3^h-1)\equiv \frac{3^m+1}{2} \pmod{3^m-1}$, where $ 1\le  h\le m-1$.    Then $\mathcal{C}_{(1,e)}$ has parameters $[3^m-1, 3^m-2m-1, 4]$ and
 is optimal if $\gcd(3^m-1, 3^h-2)=1$.
\end{theorem}
\begin{proof}
	Since $m$ is odd, we have $\gcd(\frac{3^m+1}{2}, 3^m-1)=2$ and consequently $\gcd(e(3^h-1), 3^m-1)=2$, which implies $\gcd(e, 3^m-1)=2$ and $\gcd(3^h-1,3^m-1)=2$.
 Hence $\gcd(h, m)=1$,  $e\not\in C_{1}$,   and $|C_{e}|=m$.
  According to Lemma \ref{thm-DH},   $\mathcal{C}_{(1, e)}$ has parameters $[3^m-1, 3^m-2m-1, 4]$ if the equation $(1+x)^{e}=\pm (x^e+1)$ has no solution in $\F_{3^m}\setminus \F_{3}$. Raising both sides of the above
 equation  to the  $(3^h-1)$-th power will lead to $(1+x)^{e(3^h-1)}=(x^e+1)^{3^h-1}$, i.e.,  
\begin{equation}\label{man!}
	(1+x)^s(1+x)(x^e+1)=x^{e\cdot 3^h}+1=x^{e(3^h-1)}\cdot x^e+1=x^{s}\cdot x^{1+e}+1
\end{equation}	
	 due to the fact that $e(3^h-1)\equiv 1+s\pmod{3^m-1}$, where $s=\frac{3^m-1}{2}$. 
We will discuss the solutions of equation (\ref{man!}) in the following four cases.

	Case 1), $(1+x)^s=x^s=1$. In this case, (\ref{man!}) becomes $(x^e+1)(x+1)=x^{e+1}+1$, i.e.,   $x^e+x=x(x^{e-1}+1)=0$ and consequently $x^{e-1}=-1$
due to $x\ne 0$. Then we have $x^{2(e-1)}=1$.    Note  that
\begin{eqnarray}
	&&\gcd(\frac{3^m-1}{2}, (e-1)(3^h-1))=\gcd(\frac{3^m-1}{2}, e(3^h-1)-(3^h-1))\nonumber\\
	&=&\gcd(\frac{3^m-1}{2}, \frac{3^m+1}{2}-3^h+1)=\gcd(\frac{3^m-1}{2}, 3^h-2)=1.
\end{eqnarray}
 Hence $\gcd(3^m-1, e-1)=1$ and $\gcd(3^m-1, 2(e-1))=2$, which implies $x^2=1$,  i.e.,  $x=\pm 1$.  
	
	Case 2), $(x+1)^s=-1$, $x^s=1$. In this case, $-(x+1)(x^e+1)=x^{e+1}+1$, 
i.e.,  $x^{e+1}-x^e-x+1=(x^e-1)(x-1)=0$, which implies  $ x=1$ or $x^e=1$. Since $\gcd(e, 3^m-1)=2$,  we have $x^2=1$, i.e.,   $x=\pm 1$. 
	
	Case 3), $(x+1)^s=-1$, $x^s=-1$. In this case, (\ref{man!}) becomes $-(x+1)(x^e+1)=-x^{e+1}+1$, i.e.,  $x^{e+1}+x^e+x+1-x^{e+1}+1=0$ and 
	\begin{equation}\label{out!}
		x^e+x+2=0
	\end{equation}
	Raising the both sides of (\ref{out!}) to the   $3^h$-th power will lead to 
	\begin{equation}\label{mamba}
		x^{e\cdot 3^h}+x^{3^h}+2=0
	\end{equation}
	By (\ref{out!}) and (\ref{mamba}), we have $x^{e\cdot 3^h}-x^e+x^{3^h}-x=0$. Note that
 $$x^{e\cdot 3^h}=x^e\cdot x^{e(3^h-1)}=x^{e+1+s}=-x^{e+1}.$$ 
Thus,  we have $$x^{e\cdot 3^h}-x^e+x^{3^h}-x=-x^{e+1}-x^{e}+x^{3^h}-x=-x^{e}(x+1)+x^{3^h}-x=(x-1)(x+1)+x^{3^h}-x=0$$
 due to $x^e=1-x$. Therefore,  $x^{3^h}=x-x^2+1=-(x^2-x-1)=-(x^2+2x+1)-1$, i.e.,   $(x+1)^{3^h-2}=-1$.
Then we have  $(x+1)^{2(3^h-2)}=1$. Together with  $(x+1)^{3^m-1}=1$,  we can deduce that  $(x+1)^2=1$ due to $\gcd(3^h-2, 3^m-1)=1$. Hence $x=0$ or $1$.
	
	Case 4), $(x+1)^s=1$, $x^s=-1$. In this case, (\ref{man!}) becomes $(x+1)(x^e+1)=-x^{e+1}+1$, i.e.,  $2x^{e+1}+x^e+x=0$. Let $y=\frac{1}{x}$, we have $$\frac{2}{y^{e+1}}+\frac{1}{y^e}+\frac{1}{y}=0,$$
 which is $y^e+y+2=0$. According to Case 3),  $y\in \mathbb{F}_{3}$. Hence,  $x\in \mathbb{F}_{3}$.
	
	As a result, (\ref{man!}) has no solution in $\mathbb{F}_{3^m}\setminus \mathbb{F}_{3}$. This  completes the proof.
\end{proof}


\begin{remark}
For any $ 1\le  h\le m-1$,   there exists exactly one  $e$ satisfying   the conditions of Theorem \ref{thm3h-1}. 
We list some examples as follows: 
\begin{item}
\item [1)] For  $h=1$,  $e=\frac{3^m-1}{2}+\frac{3^m+1}{4}$;
\item [2)] For $h=2$,  $e=\frac{3^m-1}{2}+\frac{3^{m+1}-1}{8}\cdot \frac{3^m+1}{4}$ when  $m\equiv 1\pmod {4}$, and 
 $e=\frac{3^{m+1}-1}{8}\cdot \frac{3^m+1}{4}$ with $m\equiv 3\pmod {4}$;
\item [3)]  For $h=3$,  $e=\frac{3^{m+1}-1}{26}\cdot \frac{3^m+1}{4}$ when $m\equiv 5\pmod {6}$,    $e=\frac{3^m-1}{2}+\frac{3^{m+2}-1}{26}\cdot \frac{3^m+1}{4}\cdot \frac{3^{m+1}-1}{8}$ when  $m\equiv 1\pmod {12}$,  and 
 $e=\frac{3^{m+2}-1}{26}\cdot \frac{3^m+1}{4}\cdot \frac{3^{m+1}-1}{8}$ when   $m\equiv 7\pmod {12}$.
\end{item} 
\end{remark}

\begin{theorem}\label{the1}
 Let $m$ be a positive integer with $\gcd (m, 6)=1$. Let $e=\frac{3^{\frac{m+3}{2}}+5}{2}$, where $m\equiv 3\pmod 4$.
 Then the ternary code  $\mathcal{C}_{(1,e)}$ has parameters $[3^{m}-1, 3^{m}-1-2m, 4]$ and is optimal if  $m\not\equiv 0\pmod {13}$.
\end{theorem}
\begin{proof}
Note that  $h=\frac{m+3}{2}$ is odd due to $m\equiv 3\pmod 4$, then by  Lemma \ref{le1},  $e\notin C_{1}$ and $|C_{e}|=m$.
  Since $h$ is odd, we have $3^{h}\equiv -1\pmod 4$.
Thus $3^{h}+5\equiv 0\pmod 4$. This shows that $e=\frac{3^{h}+5}{2}$ is even. 
Thus, Condition $1)$ in Lemma \ref{thm-DH} is satisfied. Conditions $2)$ and $3)$ in Lemma \ref{thm-DH}  are met if and only if the following equation has no solution in $\F_{3^{m}}\setminus \F_{3}$:
	\begin{equation}\label{equ1}
		(x+1)^{2e}-x^{2e}+x^{e}-1=0.
	\end{equation}
	Note that $(x+1)^{2e}=(x+1)^{3^{h}+5}=(x^{3^{h}}+1)(x^{3}+1)(x+1)^{2}$. Thus,   (\ref{equ1}) can be rewritten as
 $$-x^{3^{h}+3}+x^{3^{h}+2}+x^{3^{h}+1}-x^{3^{h}}+x^{3^{h}-1}+x^{4}-x^{3}+x^{2}+x-1+x^{\frac{3^{h}+3}{2}}=0,$$
	i.e., ,
	\begin{equation}\label{equ2}
		(x^{4}-x^{3}+x^{2}+x-1)(1-x^{3^{h}-1})-x^{\frac{3^{h}+3}{2}}(x^{\frac{3^{h}-1}{2}}-1)=0.
	\end{equation}
	Suppose $x^{\frac{3^{h}-1}{2}}=1$. From $x^{3^{m}-1}=1$, we have $x=1$
 since $\gcd(3^{m}-1,\frac{3^{h}-1}{2})=1$ due to $\gcd (h,m)=1$ and $h$ is odd.
 Therefore,  $x^{\frac{3^{h}-1}{2}}\ne 1$. Thus, (\ref{equ2}) turns to $$(x^{4}-x^{3}+x^{2}+x-1)(1+x^{\frac{3^{h}-1}{2}})+x^{\frac{3^{h}+3}{2}}=0,$$
	i.e., ,
	$$x^{\frac{3^{h}-1}{2}}(x^{4}-x^{3}-x^{2}+x-1)=-(x^{4}-x^{3}+x^{2}+x-1).$$
	If $x^{4}-x^{3}-x^{2}+x-1=0$, then $x^{4}-x^{3}+x^{2}+x-1=0$, which implies $x^{2}=0$. Therefore, 
	 $x^{4}-x^{3}-x^{2}+x-1\ne 0$.
 As a result, $$x^{\frac{3^{h}-1}{2}}=-\frac{x^{4}-x^{3}+x^{2}+x-1}{x^{4}-x^{3}-x^{2}+x-1}.$$
Hence we have 
	\begin{equation}\label{equ3}
		x^{3^{h}}=(\frac{x^{4}-x^{3}+x^{2}+x-1}{x^{4}-x^{3}-x^{2}+x-1})^{2}\cdot x=\frac{x^{9}+x^{8}+x^{4}-x^{3}+x^{2}+x}{x^{8}+x^{7}-x^{6}+x^{5}+x+1}\triangleq \frac{f(x)}{g(x)},
	\end{equation}
	where $f(x)=x^{9}+x^{8}+x^{4}-x^{3}+x^{2}+x$ and $g(x)=x^{8}+x^{7}-x^{6}+x^{5}+x+1$.
	
	Raising both sides of (\ref{equ3}) to the $3^{h}$-th power, we have

\begin{eqnarray}\label{equadd1}
	x^{3^{2h}}&=&\frac{f(x^{3^h})}{g(x^{3^h})} = \frac{x^{9\cdot 3^{h}}+x^{8\cdot 3^{h}}+x^{4\cdot 3^{h}}-x^{3\cdot 3^{h}}+x^{2\cdot 3^{h}}+x^{ 3^{h}}}{x^{8\cdot 3^{h}}+x^{7\cdot 3^{h}}-x^{6\cdot 3^{h}}+x^{5\cdot 3^{h}}+x^{ 3^{h}}+1}\nonumber\\
	&=&\frac{f^{9}(x)+f^{8}(x)g(x)+f^{4}(x)g^{5}(x)-f^{3}(x)g^{6}(x)+f^{2}(x)g^{7}(x)+f(x)g^{8}(x)}{f^{8}(x)g(x)+f^{7}(x)g^{2}(x)-f^{6}(x)g^{3}(x)+f^{5}(x)g^{4}(x)+f(x)g^{8}(x)+g^{9}(x)}\nonumber\\
	&\triangleq& \frac{h_{1}(x)}{h_{2}(x)}.
\end{eqnarray}

	Note that $x^{3^{2h}}=x^{3^{m+3}}=x^{27}$. By (\ref{equadd1}),  $x^{27}=\frac{h_{1}(x)}{h_{2}(x)}$, i.e.,  $h_{2}(x)\cdot x^{27}-h_{1}(x)=0$. Let $h(x)=h_{2}(x)\cdot x^{27}-h_{1}(x).$  	
	With Magma program, $h(x)$ can be factorized into the product of irreducible factors over $\F_{3}$ as 
	$$h(x)=x(x+1)(x-1)^{9}(x^{9}-x^{7}-x^{5}+x^{4}+x^{3}+x^{2}-1)(x^{9}-x^{7}-x^{6}-x^{5}+x^{4}+x^{2}-1)\cdot k(x),$$
	where $k(x)$ is the product of six irreducible polynomials of degree $13$ over $\F_{3}$.
	Therefore, (\ref{equ1}) has no solution in $\mathbb{F}_{3^m}\setminus \mathbb{F}_{3}$ if $m\not\equiv 0\pmod {13}$.  This  completes the proof.
\end{proof}

Similar as the proof of Theorem \ref{the1}, we have the following theorem.

\begin{theorem}\label{the2}
	Let $m$ be a positive integer with $m\equiv 1\pmod 6$. Let $e=\frac{3^{\frac{m+2}{3}}+5}{2}$. Then the ternary code  $\mathcal{C}_{(1,e)}$ has parameters $[3^{m}-1, 3^{m}-1-2m, 4]$ and is optimal.
\end{theorem}

\begin{remark}
	Theorems \ref{the1} and \ref{the2}  give  positive support of the open problem 7.9 in  \cite{ding2013optimal}.  
\end{remark}

\section{Conclusions}\label{Section5}

In this paper,  four classes of  optimal ternary   cyclic codes $\mathcal{C}_{(0,1,e)}$ and  $\mathcal{C}_{(1,e,s)}$  were constructed by
analyzing the  solutions of certain  equations over $\F_{3^m}$. 
Moreover, by analyzing the irreducible factors of certain polynomials and  the
solutions of certain equations over $\F_{3^m}$,
we presented
four classes of
optimal ternary  cyclic codes     with parameters $[3^m-1,3^m-2m-1,4]$.   
It is shown that our new optimal cyclic codes are    inequivalent to the known ones.

\end{document}